\documentclass[12pt]{article}
\usepackage{amsmath,bbm,bm,amssymb}
\usepackage{times}
\usepackage{longtable}
\usepackage{lscape}   % Optional: Rotate the table if needed
\usepackage{graphicx,psfrag,epsf,natbib,subfig}
\usepackage{subcaption}
\usepackage{booktabs}
\usepackage{multirow}
\usepackage{adjustbox}
\usepackage{wrapfig}
\usepackage{algorithmic,algorithm}
\usepackage{hyperref}
\usepackage{colortbl}
\usepackage{xcolor}
\usepackage{rotating}
\usepackage{tikz}
\graphicspath{{./image/}}
\usepackage{amsthm}
\usepackage[margin=1in]{geometry}

\let\LambdaOLD\Lambda
\renewcommand{\Lambda}{{\bm\LambdaOLD}}
\let\Psiold\Psi
\renewcommand{\Psi}{{\bm\Psiold}}
\let\Sigmaold\Sigma
\renewcommand{\Sigma}{{\bm\Sigmaold}}

\newcommand{\equ}[1]{\begin{equation} #1 \end{equation}}

\newcommand{\lm}{\bm{\Lambda}}
\newcommand{\sss}{\bm{S}}

\newcommand{\psii}{\bm{\Psi}}

\newcommand{\Mu}{\bm{\mu}}
\newcommand{\eps}{\bm{\epsilon}}

\newcommand{\yy}{\bm{y}}

\newcommand{\xx}{\bm{x}}
\newcommand{\tr}{\text{tr}}

\newtheorem{lemma}{Lemma}

\usepackage[many]{tcolorbox}    
\newtcolorbox{mybox}[1][]{
    tikznode boxed title,
    enhanced,
    arc=0mm,
    interior style={white},
    attach boxed title to top center= {yshift=-\tcboxedtitleheight/2},
    fonttitle=\bfseries,
    colbacktitle=white,coltitle=black,
    boxed title style={size=normal,colframe=white,boxrule=0pt},
    #1}

\date{}
\title{\vspace{-0.7in}A Hybrid Mixture of $t$-Factor Analyzers for Robust Clustering of High-dimensional Data}
\author{Kazeem Kareem and Fan Dai\\ {\small Department of Mathematical Sciences, Michigan Technological University, Houghton, Michigan.}}

\begin{document}
 \maketitle
\begin{abstract}
This paper develops a novel hybrid approach for estimating the mixture model of $t$-factor analyzers (MtFA) that employs multivariate $t$-distribution and factor model to cluster and characterize grouped data. The traditional estimation method for MtFA faces computational challenges, particularly in high-dimensional settings, where the eigendecomposition of large covariance matrices and the iterative nature of Expectation-Maximization (EM) algorithms lead to scalability issues. We propose a computational scheme that integrates  a profile likelihood method into the EM framework to efficiently obtain the model parameter estimates. The effectiveness of our approach is demonstrated through simulations showcasing its superior computational efficiency compared to the existing method, while preserving clustering accuracy and resilience against outliers. Our method is applied to cluster the Gamma-ray bursts, reinforcing several claims in the literature that Gamma-ray bursts have heterogeneous subpopulations and providing characterizations of the estimated groups.
\end{abstract}

\noindent%

% \tableofcontents

\section{Introduction} \label {introduction}
Cluster analysis, a fundamental technique in data analysis, aims to partition a set of objects into distinct groups (clusters) such that objects within the same cluster are more similar to each other than to those in other clusters. This unsupervised learning approach finds applications in diverse fields, including pattern recognition (e.g., character recognition \citep{lecun1998gradient}, image retrieval \citep{sivic2003efficient}), image processing (e.g., image segmentation \citep{comaniciu2002mean}), bioinformatics (e.g., gene expression analysis \citep{eisen1998cluster}, protein clustering \citep{enright2002protein}), and market research (e.g., customer segmentation \citep{wedel2000}). Clustering algorithms can be broadly categorized into two main groups: nonparametric and model-based methods. Nonparametric clustering algorithms, also known as distance-based or heuristic methods, rely on defining a distance or similarity measure between objects and grouping them based on proximity. Popular examples include k-means clustering \citep{macqueen1967}, hierarchical clustering \citep{johnson1967}, DBSCAN \citep{ester1996}, spectral clustering \citep{ng2002}, and affinity propagation \citep{frey2007}. These methods are often computationally efficient and straightforward to implement, but their performance can be sensitive to the choice of distance metric and initial parameters and are lack of probabilistic characterization of the data structure.

Model-based clustering, on the other hand, assumes that the data is generated from a mixture of underlying probability distributions, each representing a different cluster. By fitting a mixture model to the data, these approaches provide a probabilistic framework for cluster assignment and offer insights into the underlying data structure. Notable examples of model-based clustering include Gaussian mixture models \citep{mclachlan2000}, latent class analysis \citep{lazarsfeld1968}, and mixture of factor analyzers \citep{ghahramani1997}. The major tasks of model-based clustering are to estimate the parameters of the group-wise distributions and assign data points to the clusters that are most likely to have generated them, providing a probabilistic inference on clusters that allows for uncertainty estimation and model selection \citep{mclachlan2000, bishop2006, fraleyetal2002}.

Though the advances in model-based clustering have addressed various challenges, these methods can bring new issues. One advance involves robust mixture modeling to handle outliers. For instance, using heavy-tailed distributions, such as the t-distribution, is one approach to mitigate the effects of outliers \citep{mclachlan1987}. While robust to outliers, these models can be more computationally demanding in high-dimensional cases and may require careful initialization. Another advance involves Bayesian methods for model selection and parameter estimation, which provide a principled way to handle uncertainty and avoid overfitting \citep{bishop2006}. However, Bayesian methods become computationally intensive for large datasets and require careful selection of prior distributions. Given the challenges from clustering high-dimensional data, the mixture of factor analyzers (MFA) \citep{mclachlan2000} provides an alternative solution to perform model-based clustering. It combines the dimensionality reduction capabilities of the factor model \citep{Lawley_1940} with the flexibility of mixtures to capture heterogeneity in grouped data. MFA assumes that within each cluster, the high-dimensional data can be represented as a linear combination of a smaller number of latent factors, plus some noise. This allows MFA to effectively model data where the observed variables are correlated within each cluster, while also capturing the differences between clusters. 

Further, to deal with non-Gaussian data, which can significantly affect the performance of MFA, robust estimation methods have been developed, such as the mixture of $t$-distributions with factor analyzers \citep{mclachlan2007}. However, in presence of large scale matrices, the current estimation method for the mixture of $t$-factor analyzers that implements the Expectation-Maximization (EM) algorithms for both the mixture components and the factor model parameters suffers from slow convergence and turns out to be computationally inefficient. 
To address this limitation, we propose a  hybrid mixture of t-factor analyzers approach called MtFAD that embeds a matrix-free computational framework into the EM algorithm to significantly enhance computational speed without jeopardizing
accuracy of the estimation results. We further extend the model to a generalized version by allowing for cluster-specific latent space dimensionality.  

The remainder of this article is organized as follows. Section \ref{methodology} describes the mixture of t-factor analyzers  approach, including model formulation, parameter estimation and model selection, and then illustrates our estimation method, the MtFAD. Section \ref{simulations} demonstrates the effectiveness of our approach compared to the existing method through synthetic simulations. Section \ref{application} applies MtFAD to cluster the Gamma-ray burst dataset. Section \ref{conclusion} provides a summary of our work and discusses possible extensions.

\section{Methodology} \label{methodology}
\subsection{Background and Preliminaries}
 \subsubsection{  Multivariate $t$-Distribution}
Suppose a $p$-dimensional random vector $\yy$ follows the multivariate $t$-distribution $t_p(\Mu, \Sigma, \nu)$ with mean $\bm{\mu}$, scale matrix $\bm{\Sigma}$, and degrees of freedom $\nu$. Then, its probability density function is,

\begin{equation}
f_t(\bm{y}; \Mu, \Sigma, \nu) = \frac{\Gamma\left(\frac{\nu + p}{2}\right)}
{\Gamma\left(\frac{\nu}{2}\right) (\nu \pi)^{p/2} |\bm{\Sigma}|^{1/2}} 
\left( 1 + \frac{1}{\nu}(\bm{y} - \bm{\mu})^T \bm{\Sigma}^{-1} (\bm{y} - \bm{\mu}) \right)^{-\frac{\nu + p}{2}} \label{t-pdf},
\end{equation}
where $\Gamma(\cdot)$ denotes the gamma function.

The multivariate $t$-distribution described above has proven essential in robust statistical modeling \citep{kotz2004multivariate} to describe data with heavier tails compared to the multivariate normal distribution, making it particularly useful in scenarios where the Gaussian assumption is inadequate. The degrees of freedom parameter $\nu$ of a multivariate $t$-distribution governs the heaviness of the tails: as $\nu \to \infty$, the multivariate $t$-distribution converges to the multivariate normal distribution, while as $\nu \to 0$, the distribution exhibits increasingly heavy tails, approaching a form that assigns greater probability mass to extreme values. 

\subsubsection{Mixture of $t$-Factor Analyzers}\label{sub: mtfad}
Suppose $\yy_1, \yy_2,\ldots, \yy_n$ are  $p$-dimensional observed random sample from the $t$-mixture density

\begin{equation}
    f(\yy_i; \Theta) = \sum_{k=1}^K \omega_k f_t (\yy_i; \Mu_k, \Sigma_k, \nu_k),
\end{equation}

 %such that $\yy_i\sim t_p(\Mu_k, \Sigma_k, \nu_k )$,' for some $k \in  \{1, \ldots, K\}$
 $i = 1, \ldots n$, . Let $u_i\sim  \mathcal{G}\left(\frac{\nu_k}{2},\frac{\nu_k}{2} \right)$ for some $k \in  \{1, \ldots, K\}$, where $\mathcal{G}(\cdot)$ denotes the gamma distribution, and define $z_{ik}$ as the latent component indicator which equals $1$ or $0$ according as $\bm{y}_i$ belongs to the $k$th group with the grouping probability $\omega_k$. Then, it can be shown that \citep{mclachlan2007},
\equ{
\yy_i | u_i, z_{ik}=1 \sim \mathcal{N}_p(\Mu_k, \Sigma_k/u_i),\label{equ:t_distr_to_normal}
}
where $\mathcal{N}_p(\cdot)$ denotes the Gaussian distribution.

Further, the data vector can be characterized via a factor model for a lower-dimensional structure, resulting in a mixture of $t$-factor analyzers (M$t$FA) with,
\begin{equation}
 \yy_i = \Mu_k+ \Lambda_k\xx_{ik}+ \eps_{ik},
\end{equation}
where the latent factors $\xx_{ik} | z_{ik} =1 \sim t_{q}(\bm{0}, \bm{I}_{q}, \nu_k )$,
 $\eps_{ik} |z_{ik}=1 \sim t_p(\bm{0}, \Psi_k, \nu_k)$, and
the joint distribution of  $\xx_{ik}$ and $\eps_{ik}$
 must be defined in a way that aligns with the t-mixture structure assumed for the marginal distribution of the observation vector $\bm{Y_i}$. Consequently, the scale matrices of the components obtain a lower-dimensional representation as $ \Sigma_k = \Lambda_k\Lambda_k^{\top}+\Psi_k$ for each $k= 1, \ldots, K$, where $\Lambda_k$ is a $p\times q$ loading matrix with a common $q$ across all the groups and $q< \min(n, p) $ and $ p+q < (p-q)^2$, which explains the common variances shared by all the $p$ variables of the data points from the $k$th group,  and $\Psi_k$ is a $p\times p$ diagonal matrix of uniquenesses where the diagonal elements represent the unique variances for the $p$ variables of the observations from the $k$th group.
 The log-likelihood is then given by, 
\equ{ \ell= \sum_i\log \left( \sum_{k=1}^K \omega_kf_t(\yy_i;\Mu_k, \Lambda_k\Lambda_k'+ \Psi_k, \nu_k )  \right),\label{mtfa_loglik} } 
 where $f_t$ is the probability density function of a multivariate $t$-variable, as defined in \eqref{t-pdf}.

\citet{mclachlan2007} proposes an alternating expectation-conditional maximization (AECM) algorithm to estimate the parameters of MtFA which is derived as an extension of the approach for estimating mixture of  Gaussian factor analyzers \citep{mclachlanetal2003}. The method essentially merges the EM algorithm for estimating mixture components with an extra EM iteration to estimate the local factor model parameters in each of the resulting components. In the EM steps for mixture components, the data is augmented with the unobserved group indicators $z_{ik}$, and  in the EM steps for factor analyzers, the underlying factors $\xx_{ik}$ are also included in the complete data. The implementation of this AECM algorithm is given in the R package \texttt{EMMIXmfa}, which  provides applicable computational schemes to obtain robust parameter estimation for grouping data with heavy tails. However, this method suffers from slow convergence inherent in EM algorithms, and it can become significantly slow when dealing with high dimensional problems, making it a challenge for the algorithm to scale. In view of this, we propose in the following section a hybrid expectation-conditional maximization (ECM) approach that incorporates matrix-free computations into an EM framework to significantly increase computational speed with minimum memory usage while preserving high estimation accuracy. 

%%%%%%%%%%%%%

\subsection{A Hybrid ECM Algorithm for MtFA}

As per Section \ref{sub: mtfad}, the conditional distribution of $u_{i}$ is given by,
\begin{equation}
 u_{i} | \yy_i, z_{ik} = 1 \sim \mathcal{G}(m_{1k}, m_{2k}),   
\end{equation}
where
$m_{1k} = (\nu_k + p)/2$, $\quad m_{2k} = (\nu_k + \delta(\yy_i, \Mu_k; \Sigma_k))/2$, and $\delta(\yy_i, \Mu_k; \Sigma_k)= (\yy_i- \Mu_k)^{\top} [\Sigma_k]^{-1}\allowbreak(\yy_i- \Mu_k)$. Employing this result, we describe the ECM steps based on the complete or augmented data with the unobserved $u_i$ and $z_{ik}$.

First, in the expectation (E) step that deals with the unobserved or missing quantities, we compute the conditional expectation of the component indicator $z_{ik}$ for each cluster as follows,
\equ{
\gamma_{ik}=\mathbb{E}[z_{ik}|\yy_i] = \frac{\omega_k f_t(\yy_i ; \Mu_k, \Sigma_k, \nu_k)}{\sum_{k=1}^K \omega_k f_t(\yy_i ; \Mu_k, \Sigma_k, \nu_k)},
}
where again $f_t$ is the probability density function of a multivariate $t$-distribution, and $\Sigma_k=\Lambda_k\Lambda_k^{\top}+ \Psi_k$. Then, we compute the conditional expectation of $u_i$ as,
\equ{
\eta_{ik}=\mathbb{E}[u_i|\yy_i, z_{ik} = 1] = \frac{m_{1k}}{m_{2k}} = \frac{\nu_k  + p}{\nu_k + \delta(\yy_i; \Mu_k, \Sigma_k)}.
}

Next, we construct the conditional maximization (CM) steps that estimate the model parameters from maximizing the complete data log-likelihood function with the expectations computed above. The first part of the CM steps involves finding the estimators of grouping probability $\omega_k$ and the mean parameter $\Mu_k$. We evaluate these quantities by leveraging  \eqref{equ:t_distr_to_normal}, and then obtain their corresponding estimates as follows,
\begin{align}
\hat{\omega}_k &= \frac{\sum_{i=1}^n \gamma_{ik}}{n},\\
\hat{\Mu}_k &= \frac{\sum_{i=1}^n \gamma_{ik} \eta_{ik}  \yy_i}{\sum_{i=1}^n \gamma_{ik}  \eta_{ik}   }.
\end{align}
The second part of the CM steps requires to estimate the two factor model parameters $\Psi_k, \Lambda_k$, and $ \nu_k$ given $\hat{\omega}_k$ and $\hat{\Mu}_k$.  We first obtain 
\equ{\hat{\Sigma}_k = \frac{\sum_{i=1}^n \gamma_{ik} \eta_{ik} (\yy_i - \hat{\Mu}_k)(\yy_i - \hat{\Mu}_k)^{\top}}{\sum_{i=1}^n\gamma_{ik}\eta_{ik}}.}
Then, we adapt the profile likelihood method introduced by \cite{daietal2020} to jointly update $\Psi_k, \Lambda_k$ based on $\hat{\Sigma}_k$. Specifically, the method computes the first largest $q_k$ eigenvalues and the associated eigenvectors of the matrix $\Psi_k^{-1/2}\hat{\Sigma}_k\Psi_k^{-1/2} $ in order to evaluate the profile log-likelihood function adapted from \eqref{eq:profile_likelihood}. The partial decomposition utilizes the Lanczos algorithm with implicit restart proposed by \citet{baglamareichel2005} which involves only matrix-vector multiplications and turns out to be matrix-free. Next, we estimate $\Psi_k$ by optimizing the profile log-likelihood function using the limited-memory Broyden-Fletcher-Goldfarb-Shanno quasi Newton algorithm with box conditions (L-BFGS-B), which approximates the Hessian matrix using vectors and hence, is also matrix-free. The estimator of $\Lambda_k$ is then given by the relation $\hat{\Lambda}_k = \hat{\Psi}_k^{1/2}\bm{V}_{q,k}\bm{\Delta} $ as detailed in Section \ref{sec:supp-meth}.

Lastly, we determine the degrees of freedom, $\nu_k$ from solving the following equation
\begin{equation}
\label{eq:dof equation}
\begin{split}
-\psi\left(\frac{\nu_k}{2}\right) &+ \log\left(\frac{\nu_k}{2} \right) + 1 + \frac{1}{n_k} 
\sum_{i=1}^n \gamma_{ik} \left(\log \eta_{ik} - \eta_{ik}\right) + \psi\left(\frac{\nu_k+p}{2}\right) \\
&- \log\left(\frac{\nu_k+p}{2}\right)  = 0. 
\end{split}
\end{equation}
\eqref{eq:dof equation} has no known analytical solution, but  the unique root can be easily found by numerical method. Our algorithm embedded with the matrix-free computational framework is called as MtFAD (Mixture of $t$-Factor Analyzers for Data).

\subsection{Initialization and Stopping Criteria}
\label{init-stop}
A notable challenge with the EM type algorithms is the tendency to converge to local optima rather than the global maximum \citep{shiremanetal2016}. This issue arises because the EM algorithm is designed to start with a set of initial parameter values and increase the likelihood function iteratively. To mitigate the local convergence, we apply the emEM method proposed by \citet{Biernackietal2003}, which starts the algorithm with a large number of random initializations and run with a few iterations to determine several candidate initials that give the highest data log-likelihood values, then, the algorithm will run with the selected initials until convergence and the optimal estimation results come from the run with highest final data log-likelihood value. Besides, we also include an extra initialization from k-means clustering to ensure that random initials do not perform less than the data-driven one. The algorithm stops when there is no more significant increase in the data log-likelihood value with a tolerance level set to be $10^{-6}$ in practice, or when it reaches the maximum iterations of $500$.

\subsection{Model Selection}
The vast majority of the methods that focus on determining the optimal number of clusters, $K$, and the number of factors, $q$, are based on consideration of the log-likelihood function \citep{fraleyetal2002}. We select the best model by the Bayesian Information Criteria (BIC) \citep{schwarz1978}, where $\text{BIC} = -2\ln \hat{L} + k_p \ln{n} $, with $\hat{L}$ denoting the log-likelihood of the observed data, and $k_p$ representing the total number of parameters to be estimated. The optimal model would correspond to the lowest BIC. To address identifiability challenges in the model, the condition that $\Lambda_k^{\top}
\Psi_k^{-1}\Lambda_k $ is a diagonal matrix is imposed in our algorithm which also reduces the number of parameters to be estimated for each cluster by $q(q-1)/2$ \citep{lawleymaxwell1971}. Hence, we obtain the number of parameters in the BIC as $k_p =2K-1 + Kp + \sum_{k=1}^K(pq+p-q(q-1)/2)$.

%%%%%%%%% BIC is among the the easily implemented methods that has been repeatedly shown to demonstrated good performance. ( Finite mixture  models and model-based clustering by Volodymyr Melnykov and Ranjan Maitra)

\subsection{Numerical Challenges for High-dimensional Data}
One of the main difficulties in high-dimensional model-based clustering is evaluating the likelihood function for mixture models. In high dimensions, certain cluster covariance matrices can become singular, often because the number of data points in a cluster is smaller than the number of features. Consequently, computing the inverse of such covariance matrices, which is essential for the likelihood function, becomes problematic. However, in the context of  mixtures of factor analyzers, this challenge can be resolved by applying the Woodbury identity to $ \Sigma_k = \Lambda_k\Lambda_k^{\top}+\Psi_k$ so that we obtain an alternative formula
\begin{equation}
\Sigma_k^{-1}= \Psi_k^{-1} - \Psi_k^{-1}\Lambda_k(\bm{I}_{q} + \Lambda_k^{\top}\Psi_k^{-1}\Lambda_k )^{-1} \Lambda_k^{\top}\Psi_k^{-1}.
\end{equation}

The determinant is then evaluated as
\begin{equation}
\det(\Lambda_k\Lambda_k^{\top}+\Psi_k)= \frac{\det(\Psi_k)}{ \det(\bm{I}_{q} - \Lambda_k^{\top} (\Lambda_k\Lambda_k^{\top}+\Psi_k)^{-1} \Lambda_k  ) } = \frac{\det(\Psi_k)}{ \det(\bm{I}_{q}+ \Lambda_k^{\top} \Psi_k^{-1} \Lambda_k  ) }.
\end{equation}
These computations would typically be tractable provided that the $n_k$ data points in the $k$th cluster is no less than the number of factors $q$, which can generally be controlled.   

\subsection{A Generalized mixture of $t$-factor analyzers}
Further, we propose a generalized version of the current MtFA by varying the number of factors for different data groups. Consequently, the model is specified with $\xx_{ik} | z_{ik} =1 \sim t_{q_k}(\bm{0}, \bm{I}_{q_k}, \nu_k )$, and $\Lambda_k$, a $p\times q_k$ matrix, where $q_k$ is the number of factors for the $k$th group, so that $\bm{q} = (q_1,\ldots,q_K)$ is the vector of factors. This extension introduces more flexibility to facilitate the characterization of different data clusters and can be easily implemented with a modified version of our current algorithm, which is called as MtFAD-q. The generalized MtFA requires an extra constraint on $q_k$ to guarantee the model estimability as described in the following.
\begin{lemma} Suppose the number of factors for each of the $K$ data groups is $q_k$ and data vector is $p$-dimensional, then we have
$$\max_{k\in \{1 , \ldots, K\} } q_k< p + (1-\sqrt{1+8p})/2. $$ \label{lemma1}
\end{lemma}
\begin{proof}
See Section  \ref{proof_of_lemma}.
\end{proof}

\section{Performance Evaluations} \label{simulations}
\subsection{Comparison Study between MtFAD and EMMIXt}
\label{comparisions}
We first evaluate the performance of our algorithm MtFAD with comparisons to the existing AECM algorithm introduced in Section \ref{sub: mtfad}, which is referred to as EMMIXt. The grouped datasets are generated with different clustering complexities determined by a generalized overlap rate \citep{melnykovetal2012,maitraandmelnykov2010} $\bar{\omega} = 0.001,0.005, 0.01$, where a smaller $\bar{\omega}$ implies more separations between the data clusters, which is illustrated by Figure \ref{radviz:figure_mtfa}.
\begin{figure}[h!]
    \centering
    \begin{minipage}[b]{0.31\textwidth}
        \includegraphics[width=\textwidth]{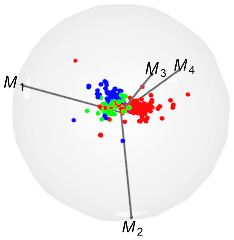}
        \caption*{$\bar{\omega}=$ 0.001 }
    \end{minipage}
    \hfill
    \begin{minipage}[b]{0.32\textwidth}
        \includegraphics[width=\textwidth]{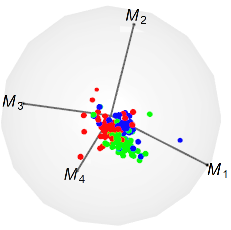}
        \caption*{$\bar{\omega}=$ 0.005 }
    \end{minipage}
    \hfill
    \begin{minipage}[b]{0.34\textwidth}
        \includegraphics[width=\textwidth]{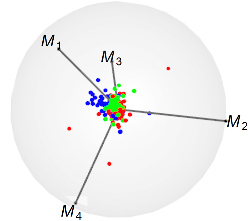}
        \caption*{$\bar{\omega}=$ 0.01 }
    \end{minipage}
    \caption{3D displays of the MtFA datasets ($n=300, p=12, K=3, \bm{q}=(2,3,4), \bm{\nu}= (2,3,3)$) with varying overlap rates. Plots were generated using the 3D radial visualization tool developed by \cite{Zhuetal2022}. }
    \label{radviz:figure_mtfa}
\end{figure}

To simulate the MtFA datasets, we first generate grouped Gaussian data $\yy^*_k$ using the R package of \texttt{MixSim} \citep{maitraandmelnykov2010} with a given $\bar{\omega}$, values of $\Mu_k$ and $\Lambda_k$ from standard normals, values of $\Psi_k$ from $\mathrm{Unif}(0.2,0.8)$, and the grouping probabilities $\omega_1,\omega_2,\ldots,\omega_K$ from standard normals and scaling the absolute values to have a sum of 1. Then, we generate random variable $s$ from the $\chi^2$ distribution with randomly generated degrees of freedom $\nu_k$, and finally, obtain the $t$-mixture data from the $k$th group as $\yy^*_k/\sqrt{s/\nu_k}+\Mu_k$ (The relation is explained in \cite{mclachlan2007}). 100 MtFA datasets were randomly generated under the settings of $n=300, p=10$ for all combinations of the true numbers of factors and clusters $q=2,3$ and $K=2,3$. Then, we applied MtFA and EMMIXt with the same initialization method and stopping rules introduced in Section \ref{init-stop} to fit the simulated datasets with $K_0$ ranging from 1 up to $2K$ and $q_0$  from 1 up to $2q$, respectively. All experiments were done using R \citep{R} on the same machine.

Table \ref{tab:mtfad BIC correctness uniform q} compares the model correctness rate that is computed as the percentage of runs at which the BIC chooses an optimal model with correct $K_0=K$ and $q_0=q$, where we can see that both MtFAD and EMMIXt exhibit almost remarkable abilities to choose the correct model. For the corrected fitted models, we evaluate the similarity between the true and estimated clusters using the adjusted rand index (ARI) \citep{melnykovetal2012}, which is shown by Figure \ref{ari_plots_mtfa} and we can see that both MtFAD and EMMIXt achieved high and almost identical clustering accuracy across all the simulation settings. Similar patterns are also shown for parameter estimations, where we evaluate the accuracy of estimates compared to the true values by the relative Frobenius distance, for example, for $\Lambda_k\Lambda_k^\top$ instead of $\Lambda_k$ due to the identifiability, the relative distance is computed as $d_{\Lambda_k\Lambda_k^\top} = \|\hat{\Lambda}_k\hat{\Lambda}_k^\top - \Lambda_k\Lambda_k^\top\|_F/\|\Lambda_k\Lambda_k^\top\|_F$, which are displayed in Figures \ref{fig:frobenius-error-lambda} and \ref{fig:frobenius-error-mu}.
\begin{figure}[h!]
    \centering
    \includegraphics[width=1\textwidth]{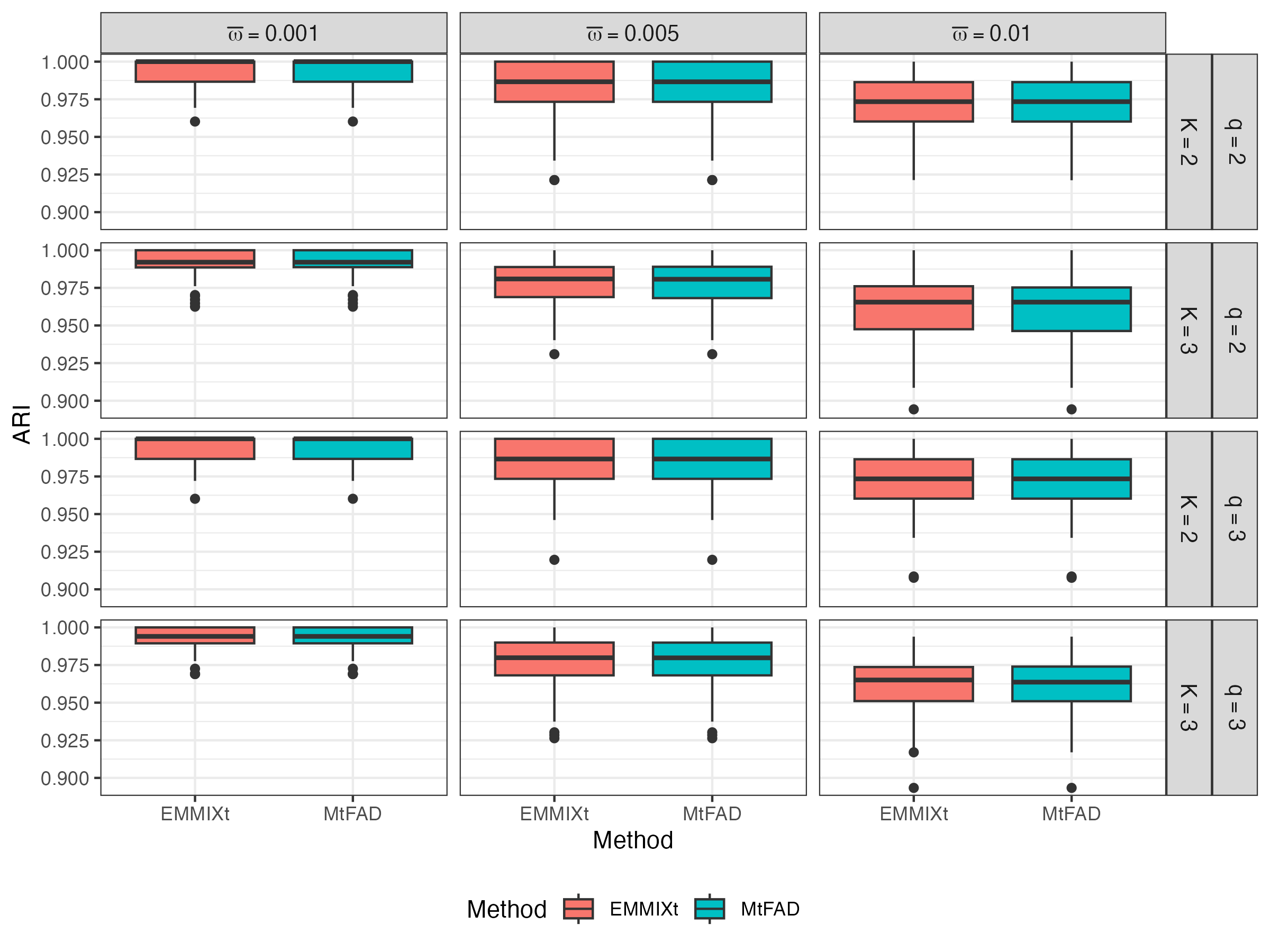} 
    \caption{Boxplots of ARI values fitted from MtFAD and EMMIXt for $n=300$ and $p=10$.}
    \label{ari_plots_mtfa}
\end{figure}

\begin{figure}[!ht]
  \centering
\mbox{
\subfloat[$K=2$]{
\includegraphics[width=0.475\linewidth]{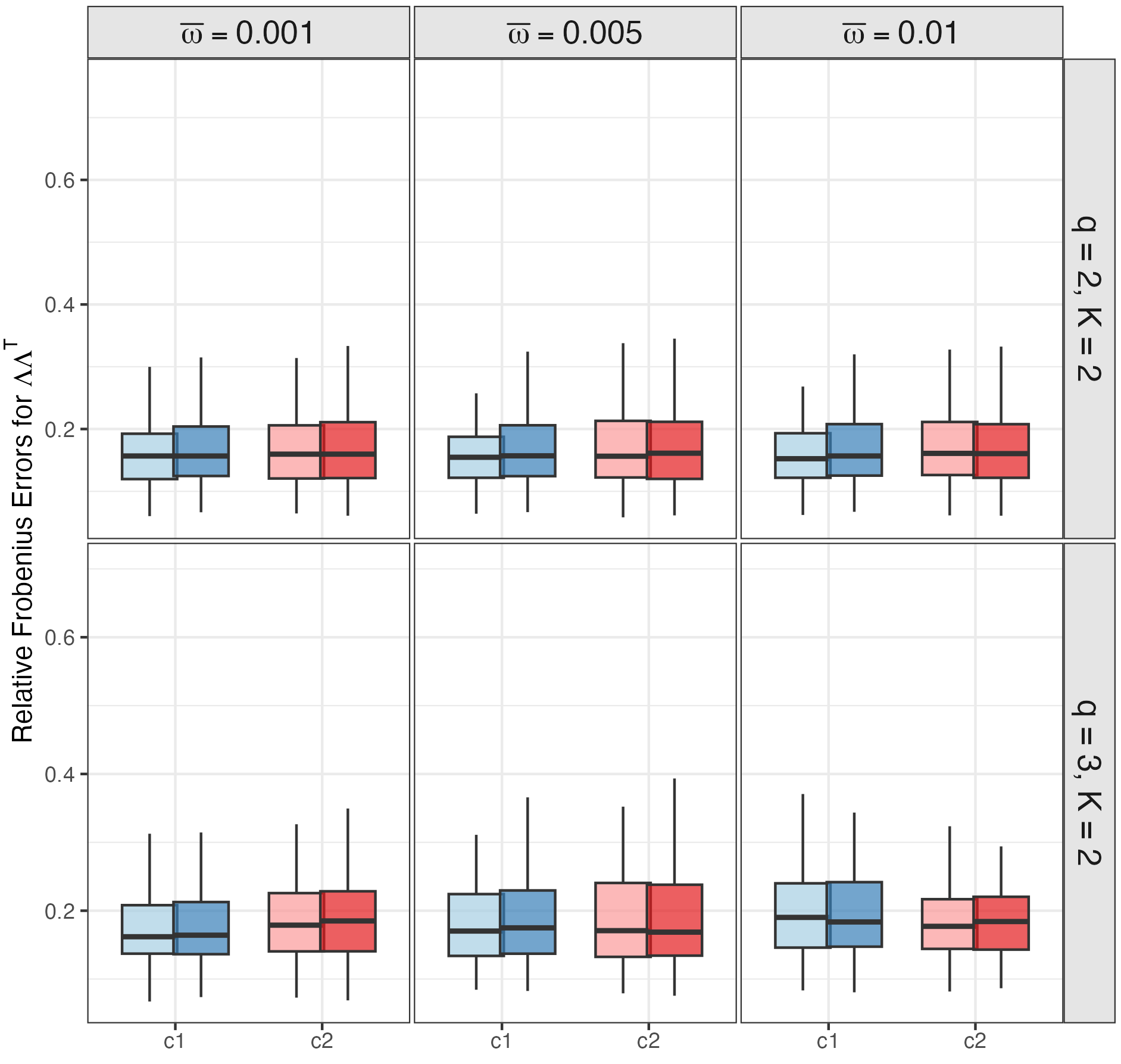}
}
}%
\mbox{
\subfloat[$K=3$]{
\includegraphics[width=0.475\linewidth]{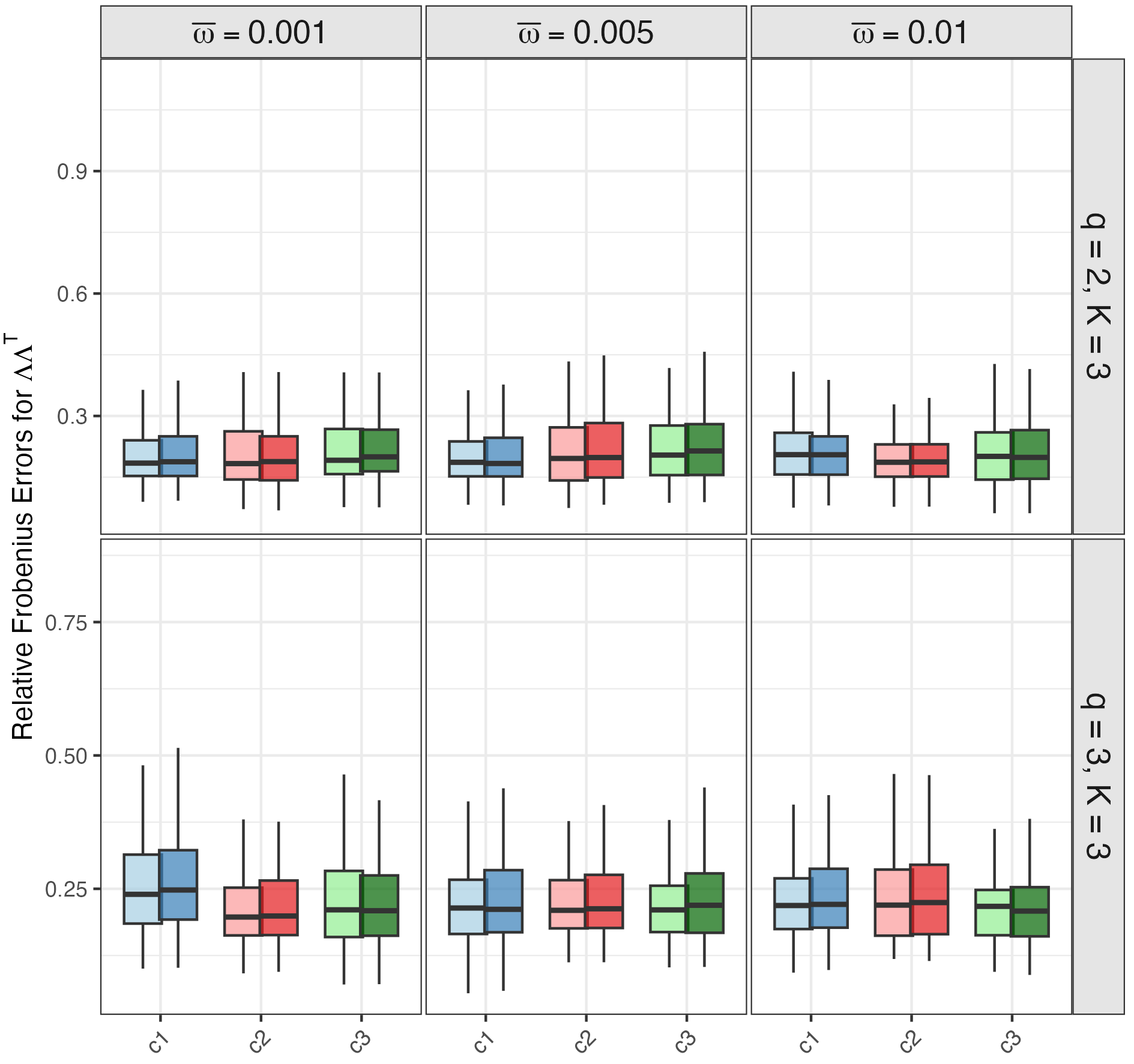}
}
}\\
\mbox{
\subfloat[$K=2$]{
\includegraphics[width=0.475\linewidth]{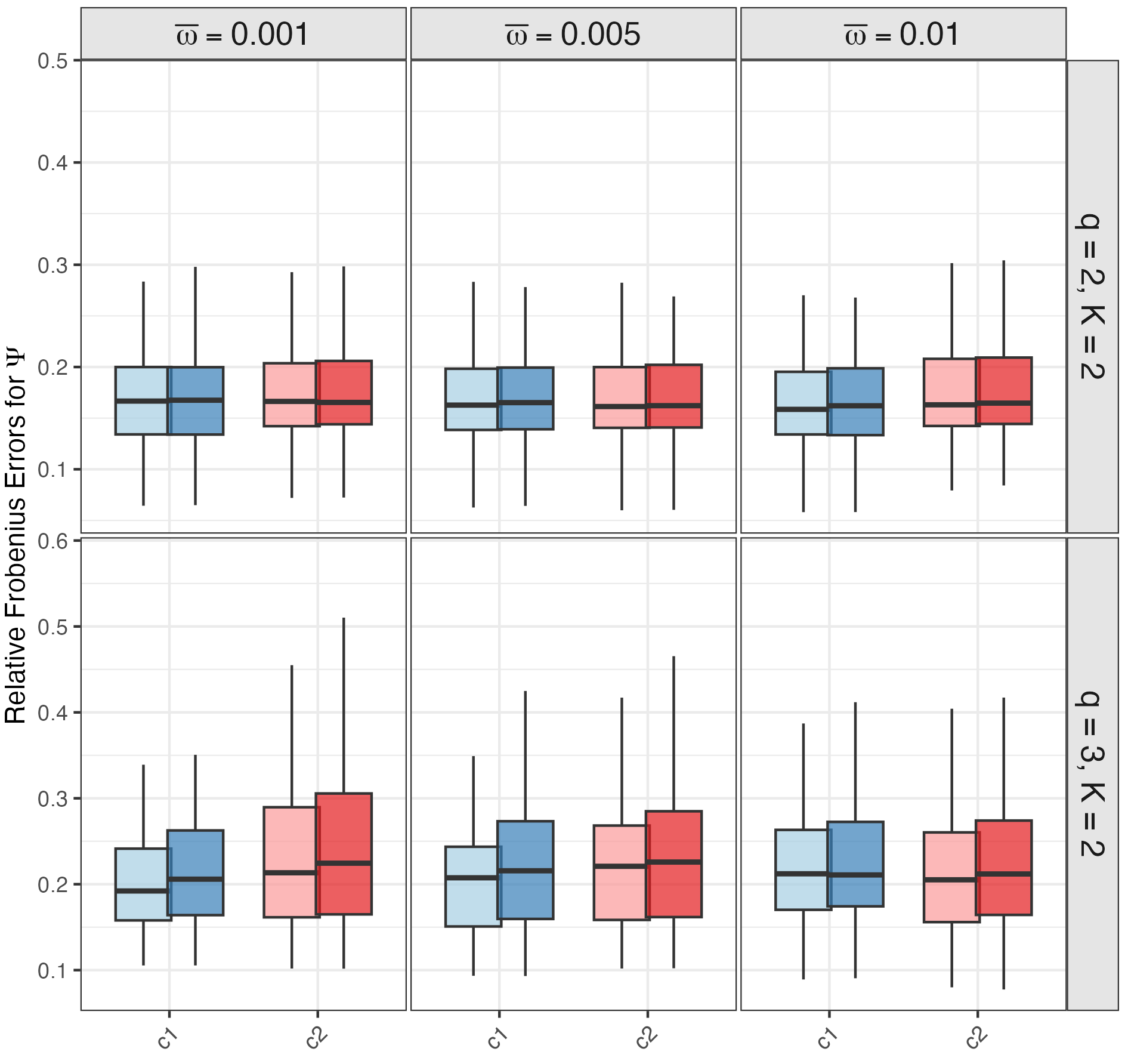}
}
}%
\mbox{
\subfloat[$K=3$]{
\includegraphics[width=0.475\linewidth]{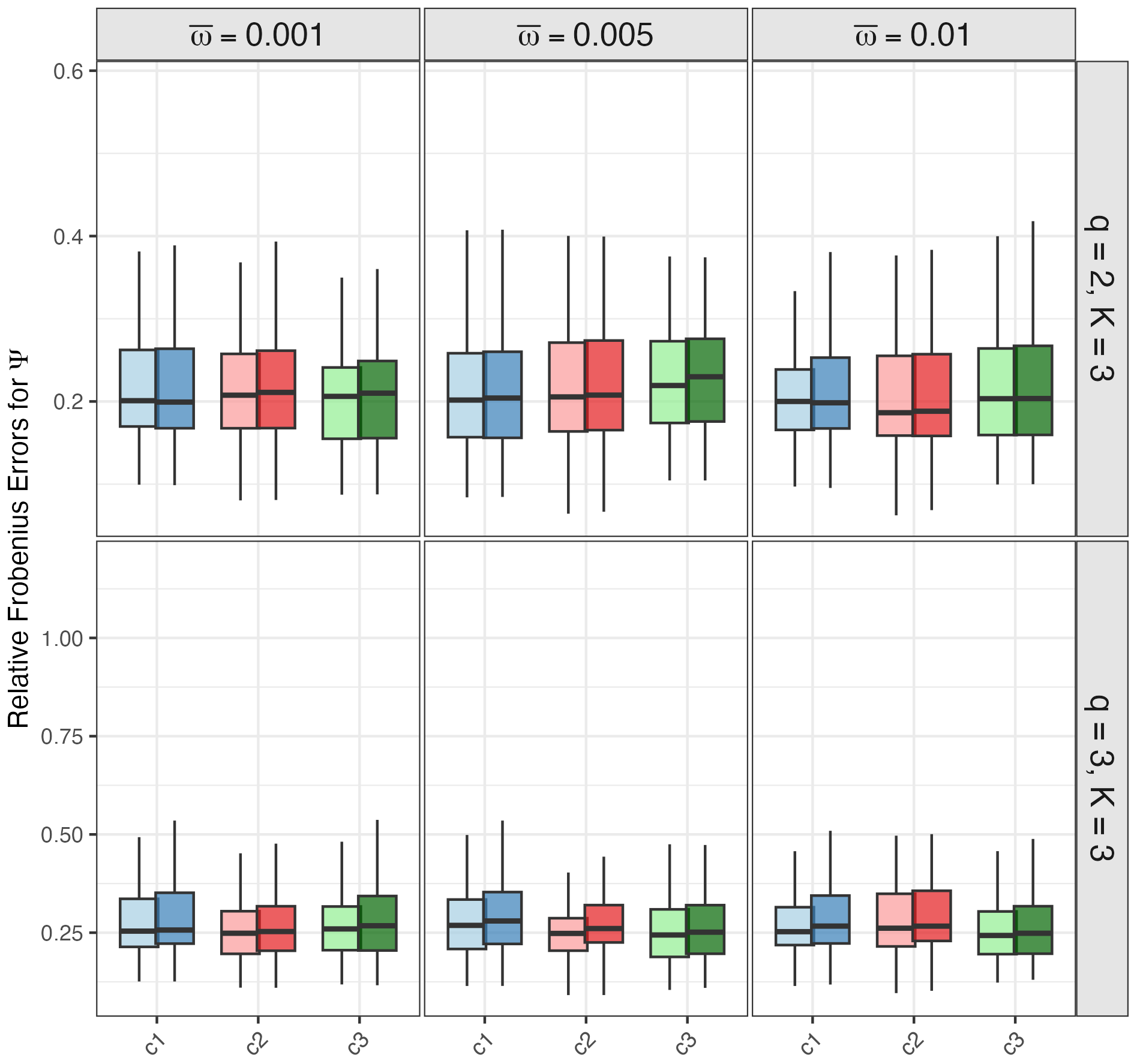}
}
}
\caption{Relative Frobenius distances  of the covariance parameters $\Lambda_k\Lambda_k^\top$ and $\Psi_k$ fitted from MtFAD and EMMIXt for $n=300$ and $p=10$ from the respective true paramters. Dark and light shades represent MtFAD and EMMIX respectively. Different colors depict different components.}
    \label{fig:frobenius-error-lambda}
\end{figure}

\begin{table}[h]
    \centering
   \caption{Correctness rates among the 100 simulation runs fitted with MtFAD and EMMIXt for $n=300,p=10$.}
    \begin{tabular}{c c  c c  c c  c c}
        \toprule\toprule%\hline\hline
        \multicolumn{2}{c}{} & \multicolumn{2}{c}{$\bar{\omega} = 0.001$} & \multicolumn{2}{c}{$\bar{\omega} = 0.005$} & \multicolumn{2}{c}{$\bar{\omega} = 0.01$} \\
        $K$ & $q$ & MtFAD & EMMIXt & MtFAD & EMMIXt & MtFAD & EMMIXt \\
        \midrule
        2 & 2 & 0.98 & 0.98 & 0.97 & 0.98 & 0.99 & 0.99 \\
        2 & 3 & 0.98 & 0.99 & 1 & 1 & 1 & 1 \\
        3 & 2 & 0.92 & 0.98 & 1 & 1 & 0.98 & 1 \\
        3 & 3 & 1 & 1 & 0.99 & 1 & 1 & 1 \\
        \bottomrule
    \end{tabular}
    \label{tab:mtfad BIC correctness uniform q}
\end{table}
With the clustering and estimation accuracies being preserved, our MtFAD outperforms EMMIXt in terms of the computational time. Figure  \ref{time_speedup_mtfa-p10} shows the relative speed of MtFAD to EMMIXt for the correct models, which is computed as the ratio of the CPU runtime of EMMIXt to that of MtFAD, where we can see that MtFAD exhibits significant time speedup compared to EMMIXt. To further investigate the performance of MtFAD for high-dimensional data, we simulated MtFA datasets with $n=p=150$. For convenience, the overlap rate is not prefixed in this case. Figure \ref{fig:runtimes-speedup-np150-mtfa} displays the runtime and relative speed of MtFAD to EMMIXt under this high-dimensional setting, where MtFAD shows higher efficiency with increasing $p$ and meanwhile, maintains a desirable clustering and estimation accuracies as presented in Figure \ref{frobenius-error-p150}.

\begin{figure}[h!]
    \centering
\includegraphics[width=0.8\textwidth]{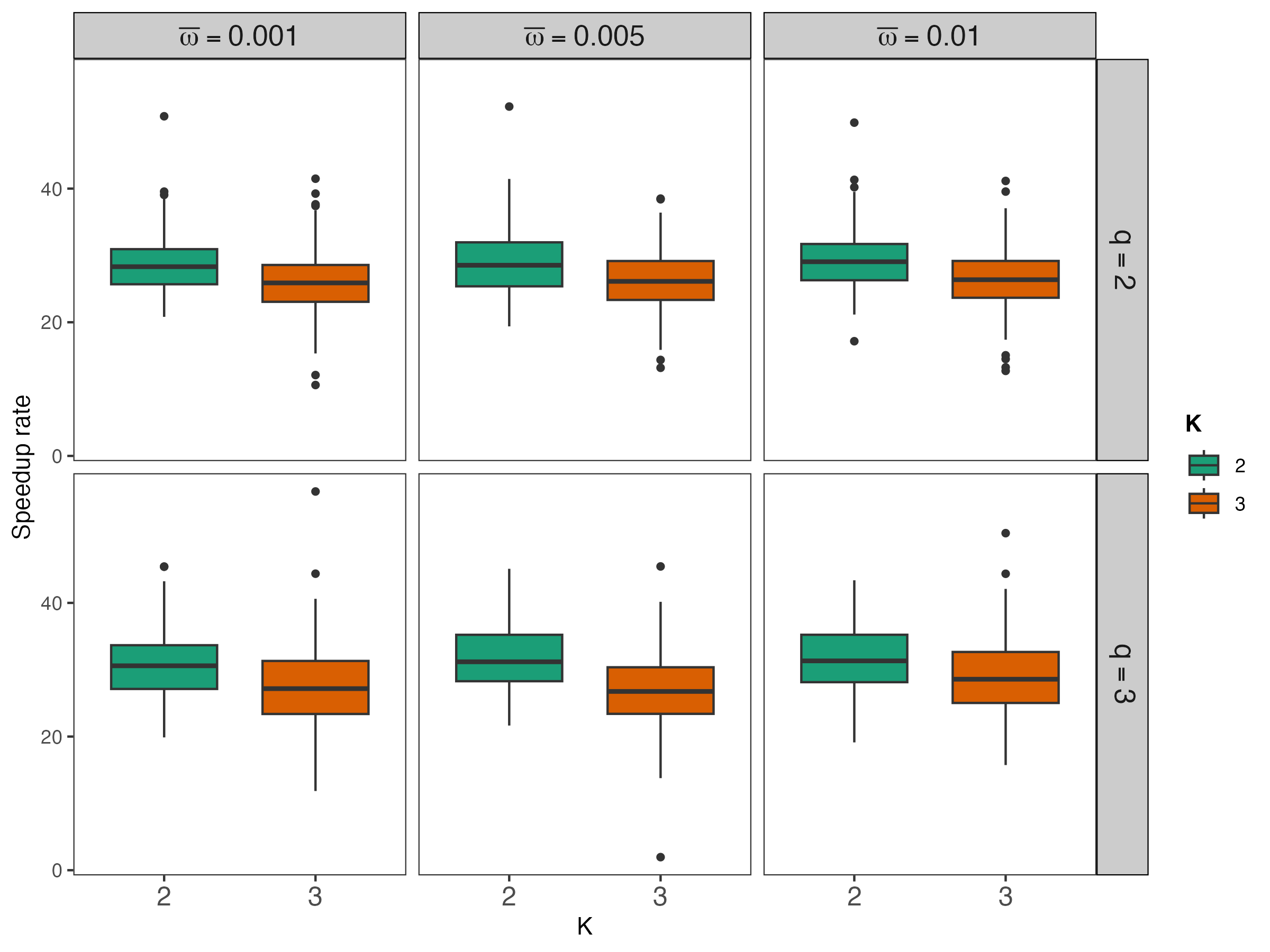} 
    \caption{Boxplots of the time speedup of M$t$FAD relative to EMMIX$t$ for $n=300,p=10$.}
    \label{time_speedup_mtfa-p10}
\end{figure}

\begin{figure}[!ht]
  \centering
\mbox{
\subfloat[Runtime for MtFAD and EMMIXt]{
\includegraphics[width=0.485\linewidth]{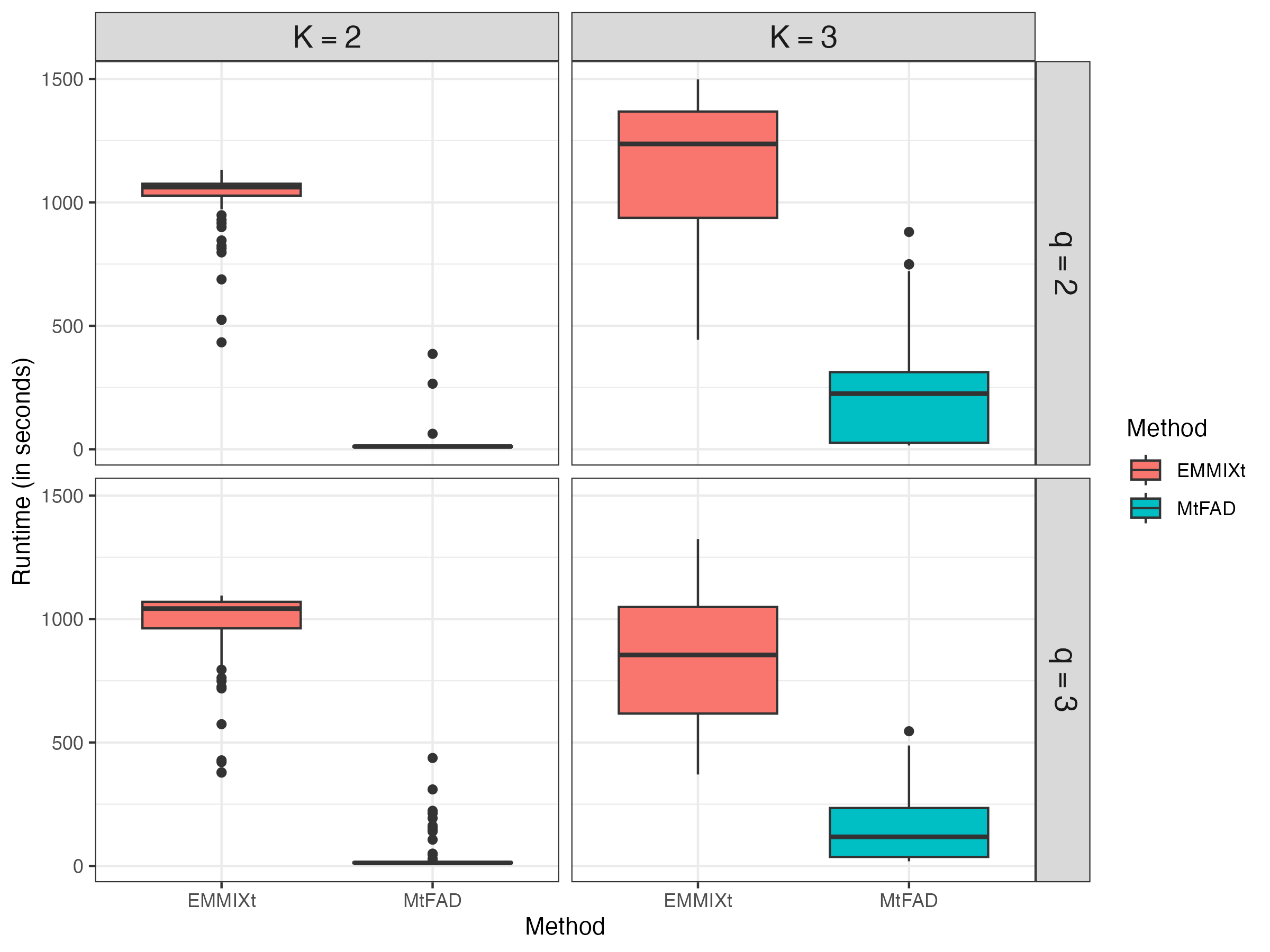}
}
}%
\mbox{
\subfloat[Relative speed of MtFAD to EMMIXt]{
\includegraphics[width=0.485\linewidth]{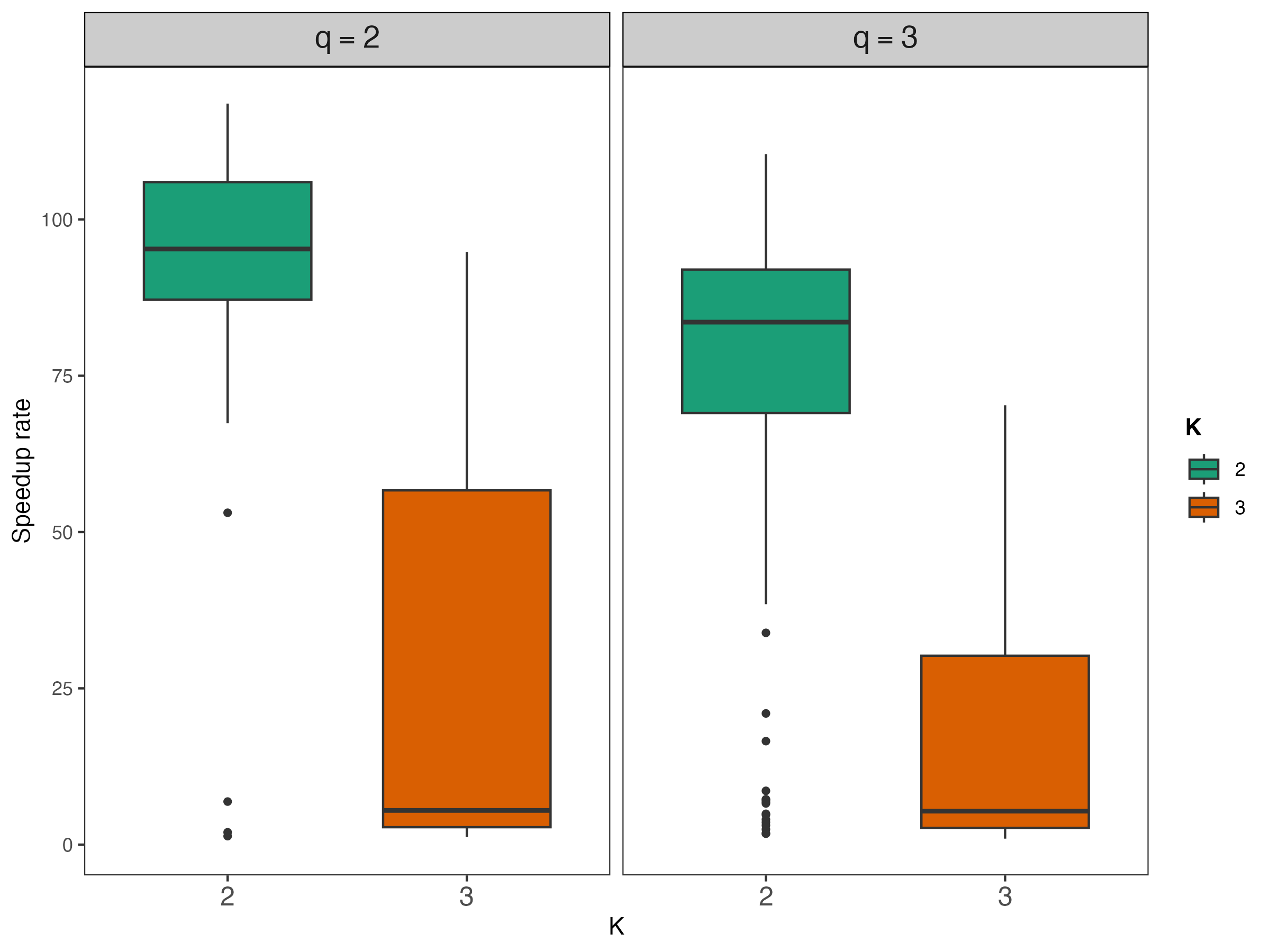}
}
}
\caption{Plots of (a) runtime and (b) relative speed of MtFAD compared to EMMIXt for data with $n=p=150$ and all combinations of $q=2,3$ and $K=2,3$. }
    \label{fig:runtimes-speedup-np150-mtfa}
\end{figure}

\subsection{Simulation Study for MtFAD-q}
 In the case of varied number of factors for different clusters, we note that the correct model is invariant with respect to the permutation of the entries of $\bm{q} = (q_1, \ldots, q_{K})$ and the corresponding permutation of the entries of $\bm{\nu} = (\nu_1,  \ldots, \nu_{K})$, because a permutation of both the entries of $\bm{q}$ and $\bm{\nu}$ can be considered a consequence of the same permutation of the order of the components in a $t$-mixture population. Therefore, the correctness rate for MtFAD-q is defined as the proportion of the runs that the BIC chooses the optimal number of components as the true $K$ and the optimal number of factors as the true $\bm{q}$, up to a permutation of the entries of $\bm{q}$.  We generated 100 samples from the generalized mixture of $t$-factor analyzers with $n=300$, $p=10$, $K = 2$, and $\bm{q} \in$ $\{ (2,3), (2,4) \} $ following the similar strategies used in Section \ref{comparisions}.
 Figure \ref{fig:correctness-mtfad-q} gives the correctness rates and it can be observed that as the clustering complexity ($\bar{\omega}$) increases, MtFAD-q maintains the ability to produce the optimal fitting at the correct model across various kinds of settings. 
\begin{figure}[h!]
    \centering
    \includegraphics[width=0.95\linewidth]{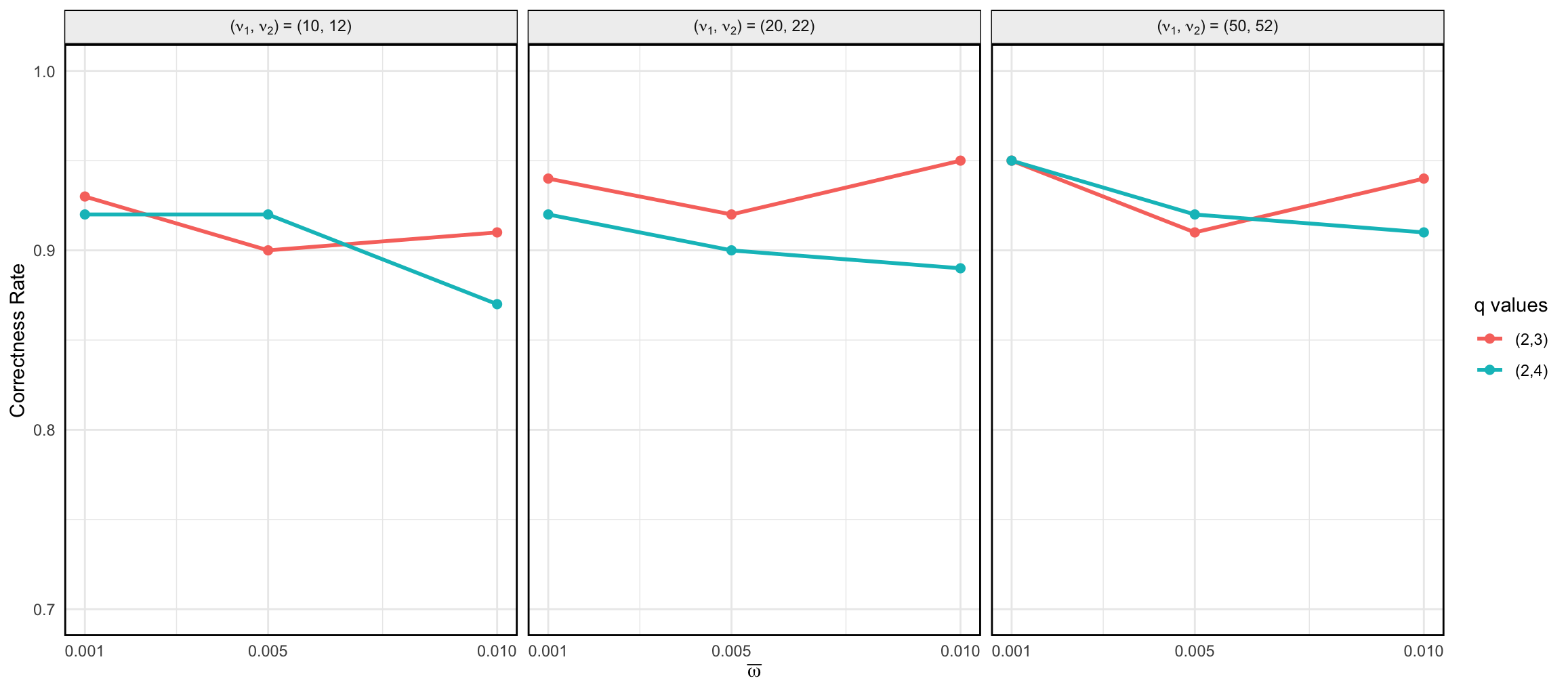}
    \caption{Correctness rates of MtFAD-q for generalized MtFA simulations.}
    \label{fig:correctness-mtfad-q}
\end{figure}

\section{Clustering Gamma-ray Bursts } \label{application}
Gamma-ray bursts (GRBs) are among the most energetic explosions in the universe, releasing vast amounts of energy in short durations, typically ranging from milliseconds to several minutes \citep{piran2004}. They are traditionally classified into long and short duration groups, with the former linked to the collapse of massive stars and the latter associated with compact object mergers \citep{berger2014}. The clustering of GRBs can provide critical insights into high-energy astrophysical processes and the evolution of the universe \citep{kumar2015}, and recent studies have argued that GRBs belong to distinct groups with more than two subpopulations \citep{shahmoradietal2015,chattopadhyayetal2017,chattopadhyayetal2018}. We consider the GRB dataset from the Burst and Transient Source Experiment (BATSE) 4Br catalog studied by \cite{chattopadhyayetal2017}. The observed 1599 GRBs were measured with nine astrophysical features: $T_{50}$ and $T_{90}$, the times by which $50\% $ and $90\%$ of the flux arrive; $P_{64}$, $P_{256}$, and $P_{1024}$, the peak fluxes measured in bins of $64$, $256$, and $1024$; $F_{1}$, $F_{2}$, $F_{3}$, and $F_4$, the four time-integrated fluences in the spectral channels of $20-50$ keV, $50-100$ keV, $100-300$ keV and $>300$ keV. 

Figure \ref{corr_matrix_grb} displays the density and correlation plots of the data which reveals the presence of extreme values for individual features. 
\begin{figure}[h!]
    \centering
           \includegraphics[width=1\textwidth]{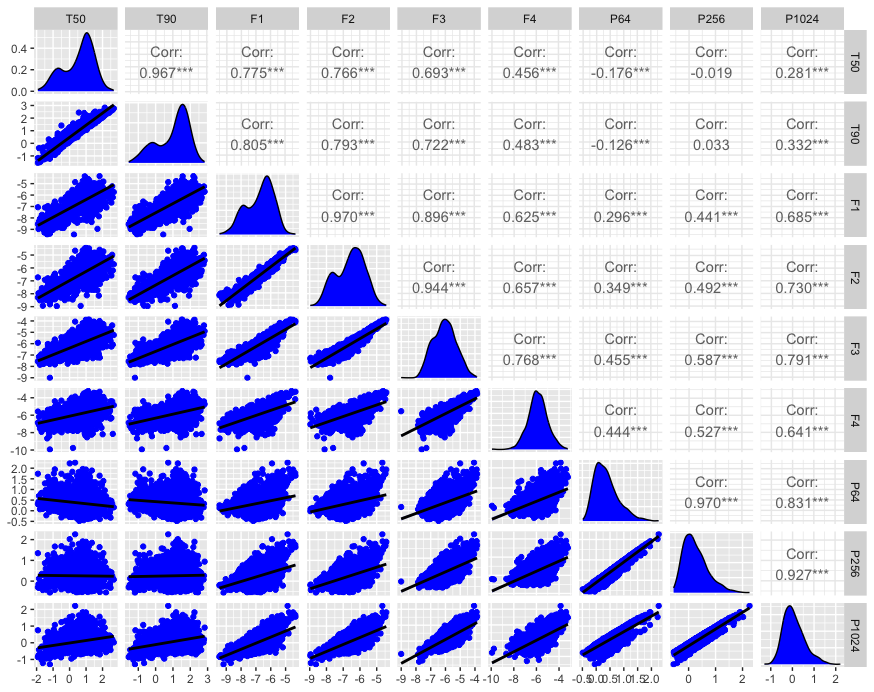}
\caption{Density and correlation plots of the nine features of the 1599 GRB records.}
    \label{corr_matrix_grb}
\end{figure} 
Hence, we employed the MtFA and fitted the data using MtFAD with common number of factors up to $5$ (the maximum number of factors given $p=9$), and with the number of groups up to $10$. The best model chosen by MtFAD via BIC has $K=5$ cluster groups with $q=4$ factors, and the group sizes are given in Table \ref{table:rgb_number_of_observations_in_groups} where there are two large groups containing over $450$ GRBs and the smallest one has less than $200$ records.
\begin{table}[h!]
    \centering
    \caption{Number of observations of the five estimated groups from MtFAD.}  
    \begin{tabular}{cccccc}
        \hline\hline
         Cluster groups & 1 & 2 & 3  & 4 & 5  \\ \hline
                Group size &  454 & 190 & 250 & 492& 213  \\ \hline

    \end{tabular}  \label{table:rgb_number_of_observations_in_groups}
\end{table}
We then applied MtFAD-q to the data without the constraint of the same $q$ for all groups. Table \ref{table: mtfad vs mtfad-q} compares the metrics of evaluating the best fitted models between MtFAD and MtFAD-q, where the the generalized model (MtFAD-q) gives better results with 4 factors detected for groups 1,2 and 5, and 5 factors for groups 3 and 4, with the number of GRBs in each group presented in Table \ref{table:rgb_number_of_observations_in_groups_varied_q}.
\begin{table}[h!]
    \centering
    \caption{Best model results fitted from MtFAD and MtFAD-q.}
    \begin{tabular}{cccccc}
        \hline\hline
    Model        & BIC        & log-likelihood & optimal $K$  & optimal $q$ or $\bm{q}$  \\ \hline
    MtFAD        &  $-3262.7$ & 2549.8         & 5            &  $q=4$                   \\ \hline
    MtFAD-q      & $-3303.1$  & 2606.9         & 5            & $\bm{q}=( 4,4,5,5,4)$    \\ \hline
    \end{tabular}
    \label{table: mtfad vs mtfad-q}
\end{table}

\begin{table}[h!]
    \centering
    \caption{Number of observations of the five estimated groups from MtFAD-q. } 
    \begin{tabular}{cccccc}
        \hline\hline
         Cluster groups & 1 & 2 & 3  & 4 & 5  \\ \hline
                Group size &  428 & 190 & 265 & 503 & 213  \\ \hline

    \end{tabular}
\label{table:rgb_number_of_observations_in_groups_varied_q}
\end{table}

Further, Table \ref{tab:loadings_grb_varied_q} shows the estimated factor loadings ($\hat{\Lambda}_k$) for each of the five components from MtFAD-q. One striking feature across the factor loading values is that the leading factors all have negative weights, implying that the factor is influenced negatively by all the features. The rest of the loadings are predominantly contrasts of positive and negative influences from the features which are distinguished between groups, indicating different kinds of characterizations of the five clusters via the factor model. Additionally, it can be seen from Figure \ref{corr_matrix_grb_colored} that there are possibly one larger cluster of groups 2 and 4 that have lower levels of $T_{50}$ and $T_{90}$ and thus represent the short duration category, and one super clusters with the rest three groups with long duration features, which confirms the traditional clustering results of GRBs, and furthermore, indicates the potential existence of subgroups within the two duration categories in terms of the peak flux and time-integrated fluence features.
\begin{figure}  
    \centering
\includegraphics[width=1\textwidth]{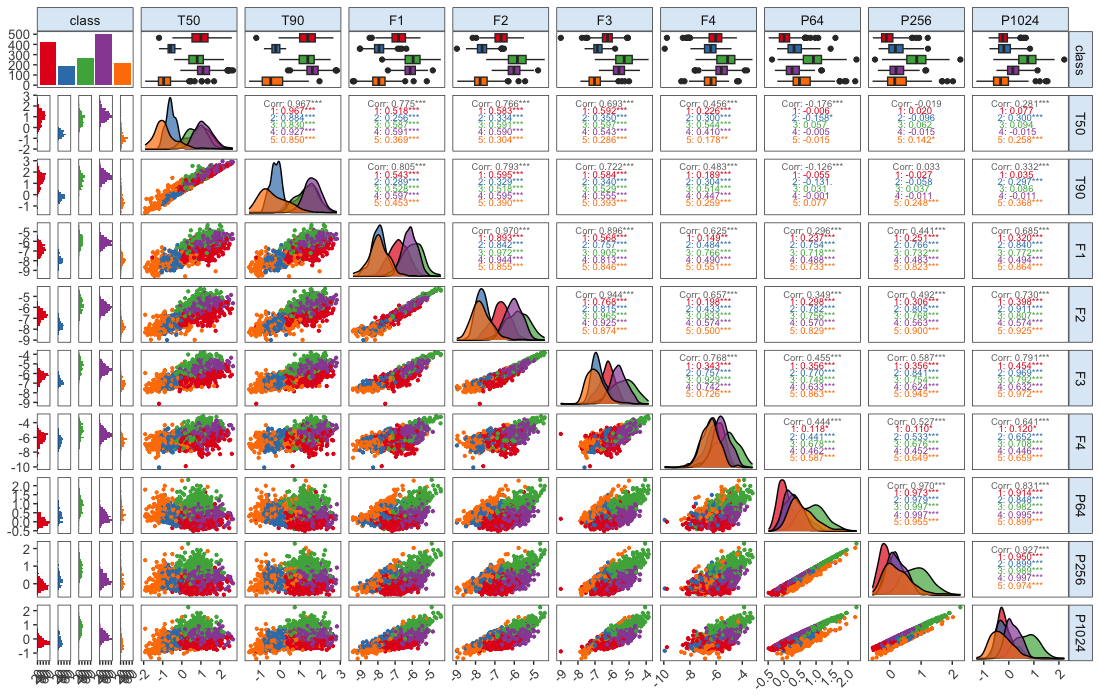}  % Use 
    \caption{Density and correlation plots of the nine features of GRB data colored by the estimated clusters from MtFAD-q.}
\label{corr_matrix_grb_colored}
\end{figure}

\section{Conclusion} \label{conclusion}
In this paper, we propose a hybrid likelihood approach for estimating the parameters from the mixture of $t$-factor analyzers with both high computational efficiency and clustering accuracy. 
The method incorporate matrix-free algorithms for likelihood optimization into the EM scheme for mixture components. The developed algorithm called MtFAD significantly accelerates the convergence of estimation results particularly for high-dimensional data compared to the existing method and exhibits remarkable clustering performance. Based on MtFAD, we further generalize the mixture model with varying numbers of factors and develop the algorithm called MtFAD-q, which allows for more flexibility for grouped data characterization. Our methods were applied to cluster the Gamma-ray burst data, identifying five kinds of GRBs that provides new insights into the subpopulations of GRBs and their astrophysical features. The proposed methods and algorithms paves the way to clustering and characterizing more complicated data such as partially recorded measurements with missing values, or time series data with correlated latent factors.

\subsection*{Supplementary Materials}\label{supplementary}
\textbf{Supplement:} Provides additional details on the algorithms, performance evaluations and data application described in this paper. Documented codes for reproducing the results are also included.

%\subsection*{Acknowledgements}
%This research is supported by the Michigan Technological University Doctoral Finishing Fellowship Award. The authors would like to appreciate the Graduate school for providing this support. 


\bibliographystyle{asadoi}
\bibliography{references}

\newpage
 \pagebreak
\begin{center}
\textbf{\large Supplement to \\
``A Hybrid Mixture of $t$-Factor Analyzers for Clustering High-dimensional Data" %\\Kazeem Kareem and Fan Dai
}
\end{center}
%%%%%%%%%% Merge with supplemental materials %%%%%%%%%%
%%%%%%%%%% Prefix a "S" to all equations, figures, tables and reset the counter %%%%%%%%%%
\setcounter{equation}{0}
\setcounter{figure}{0}
\setcounter{table}{0}
\setcounter{page}{1}
\setcounter{section}{0}
\makeatletter
\renewcommand{\thesection}{S\arabic{section}}
\renewcommand{\thesubsection}{\thesection.\arabic{subsection}}
\renewcommand{\theequation}{S\arabic{equation}}
\renewcommand{\thefigure}{S\arabic{figure}}
\renewcommand{\bibnumfmt}[1]{[S#1]}
\renewcommand{\citenumfont}[1]{S#1}
\renewcommand{\thetable}{S\arabic{table}}

%%%%%%%%%% Prefix a "S" to all equations, figures, tables and reset the counter %%%%%%%%%%
\section{Supplementary Materials for Methodology}
\label{sec:supp-meth}
\subsection{Profile Likelihood Method for Factor Model}
\label{sec:supp-profile}
\citet{daietal2020} propose the following profile method for estimating the covariance parameters in a Gaussian factor model. Given the data matrix
    $ \bm{Y}= (\yy_1, \ldots, \yy_n)^{\top}$, let $ \bar{\bm{Y}} = \bm{Y}^{\top} \bm{1} / n $, 
   $\bm{W}= n^{-1/2}( \bm{Y}- \bm{1}\bar{\bm{Y}}^{\top})^{\top} \Psi^{-1/2} $,  and  $\sss= ( \bm{Y}- \bm{1}\bar{\bm{Y}}^{\top})^{\top}( \bm{Y}- \bm{1}\bar{\bm{Y}}^{\top})/n$.  Suppose the $q$ largest singular values of $\bm{W}$ are $\sqrt{\theta_1}\geq\sqrt{\theta_2}\geq \cdots \geq \sqrt{\theta_q}$, and the corresponding right-singular vectors are the columns of $\bm{V}_q$. Assume a factor model with $q$ latent factors for the data where the covariance matrix $\Sigma = \Lambda\Lambda^\top+\Psi$, then, $\lm$ can be profiled out from the data likelihood function using its maximum value $\hat{\lm} = \psii^{1/2}\bm{V}_q\bm{\Delta}$, and the profile log-likelihood function is given by
\equ{  \ell_p(\psii)=- \frac{n}{2} \left(p \ln(2\pi) + \ln|\psii| + \tr\left[ \psii^{-1} \bm{S} \right] + \sum_{i=1}^q (\ln\theta_i-\theta_i+1)
 \right),\label{eq:profile_likelihood} }
where  $\bm{\Delta}$ is a $q\times q$ diagonal matrix with $i$th  diagonal entry  $ \max(\theta_i-1,0)^{1/2}. $ 
 Here, the parameters to be estimated are  $\hat{\psii}=  \text{argmax}  \hspace{0.2cm}\ell_p(\psii) $ and then $\hat{\lm} = \hat{\psii}^{1/2}\bm{V}_q\bm{\Delta} $.

\subsection{Proof of Lemma \ref{lemma1} } \label{proof_of_lemma}

The condition that $ \Lambda_k^{\top}\Psi_k^{-1}\Lambda_k $  be diagonal imposes $\frac{1}{2}q_k(q_k-1)$ constraints on the parameters \citep{lawleymaxwell1971}.  Hence, for each $k \in \{1, \ldots, K\}$, the number of free parameters in the factor analytic model of the $k$th component is 
\begin{equation}
pq_k+p -\frac{1}{2}q_k(q_k-1) . 
\end{equation}
Suppose $s_k$ is  the difference between the number of parameters for $\Sigma_k$ and the number of free parameters considering the assumption $ \Sigma_k = \Lambda_k\Lambda_k^{\top}+\Psi_k$. Then for each $k\in \{1, \ldots K\}$, 

\begin{align}
  s_k &=  \frac{1}{2}p(p-1) - \left(pq_k + p -\frac{1}{2}q_k(q_k-1) \right)\\
      &=  \frac{1}{2}\left[ (p-q_k)^2- (p+q_k) \right]
\end{align}
This difference represents the reduction in the number of parameters for $\Sigma_k$. For this difference to be positive, each $q_k$ needs to be small enough so that

\begin{align*}
       &\frac{1}{2}\left[ (p-q_k)^2- (p+q_k) \right] > 0  \quad    \forall k \\
       \implies& \quad q_k <  p + (1-\sqrt{1+8p})/2   \quad \forall k \\
       \implies& \max_{k\in \{1 , \ldots, K\} } q_k <  p + (1-\sqrt{1+8p})/2.
\end{align*}

\section{Additional Results for Performance Evaluations}
\begin{figure}[!ht]
  \centering
\mbox{
\subfloat[$K=2$]{
\includegraphics[width=0.475\linewidth]{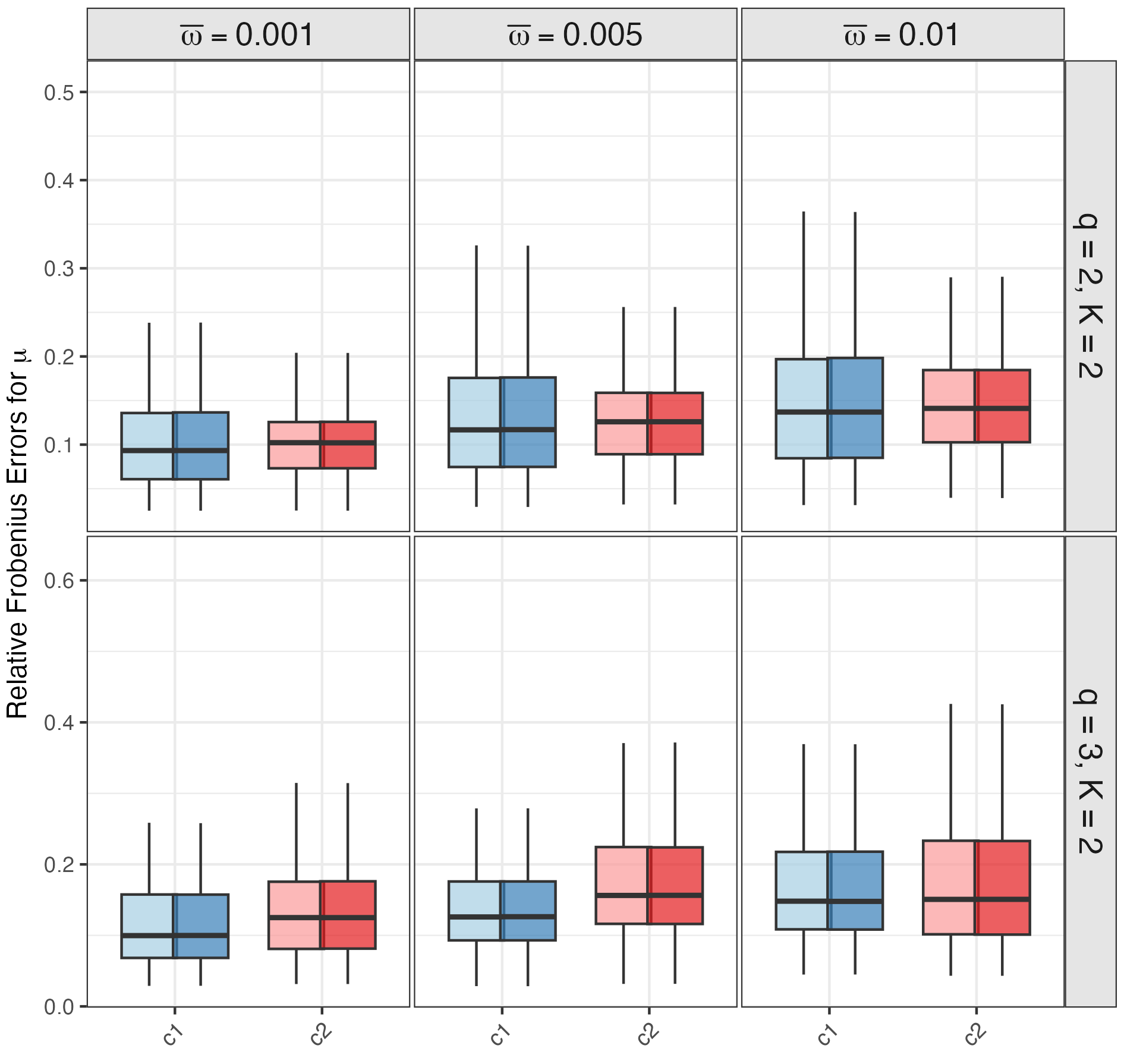}
}
}%
\mbox{
\subfloat[$K=3$]{
\includegraphics[width=0.475\linewidth]{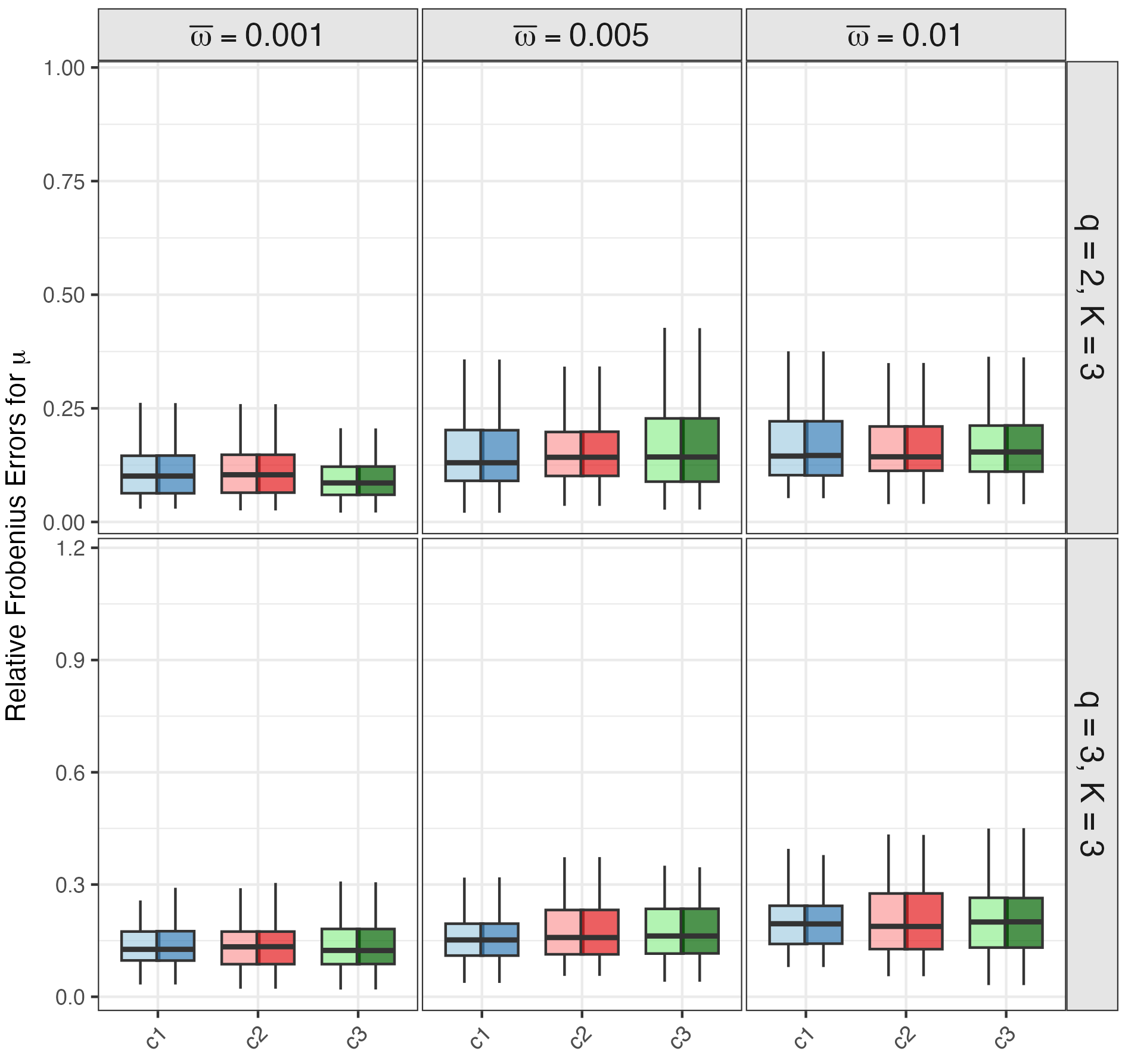}
}
}
\caption{Relative Frobenius distances  of the mean parameter $\Mu_k$ fitted from MtFAD and EMMIXt for $n=300$ and $p=10$ from the respective true paramters. Dark and light shades represent MtFAD and EMMIX respectively. Different colors depict different components.}
    \label{fig:frobenius-error-mu}
\end{figure}

\begin{figure}[h!]
    \centering
    \begin{minipage}[b]{0.51\textwidth}
        \includegraphics[width=\textwidth]{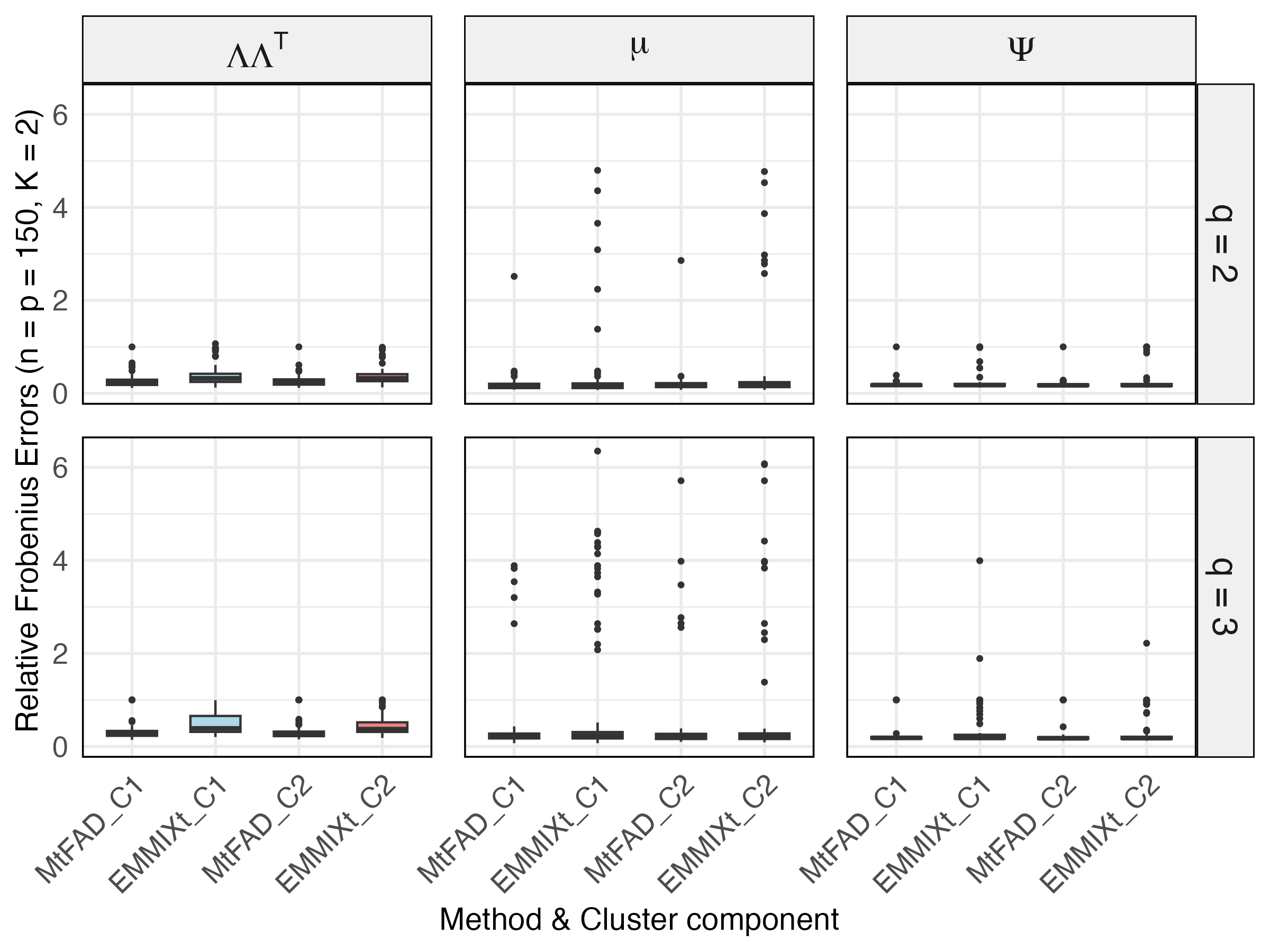}
        \caption*{(a) $K=2$ }
    \end{minipage}
    \hfill
    \begin{minipage}[b]{0.48\textwidth}
        \includegraphics[width=\textwidth]{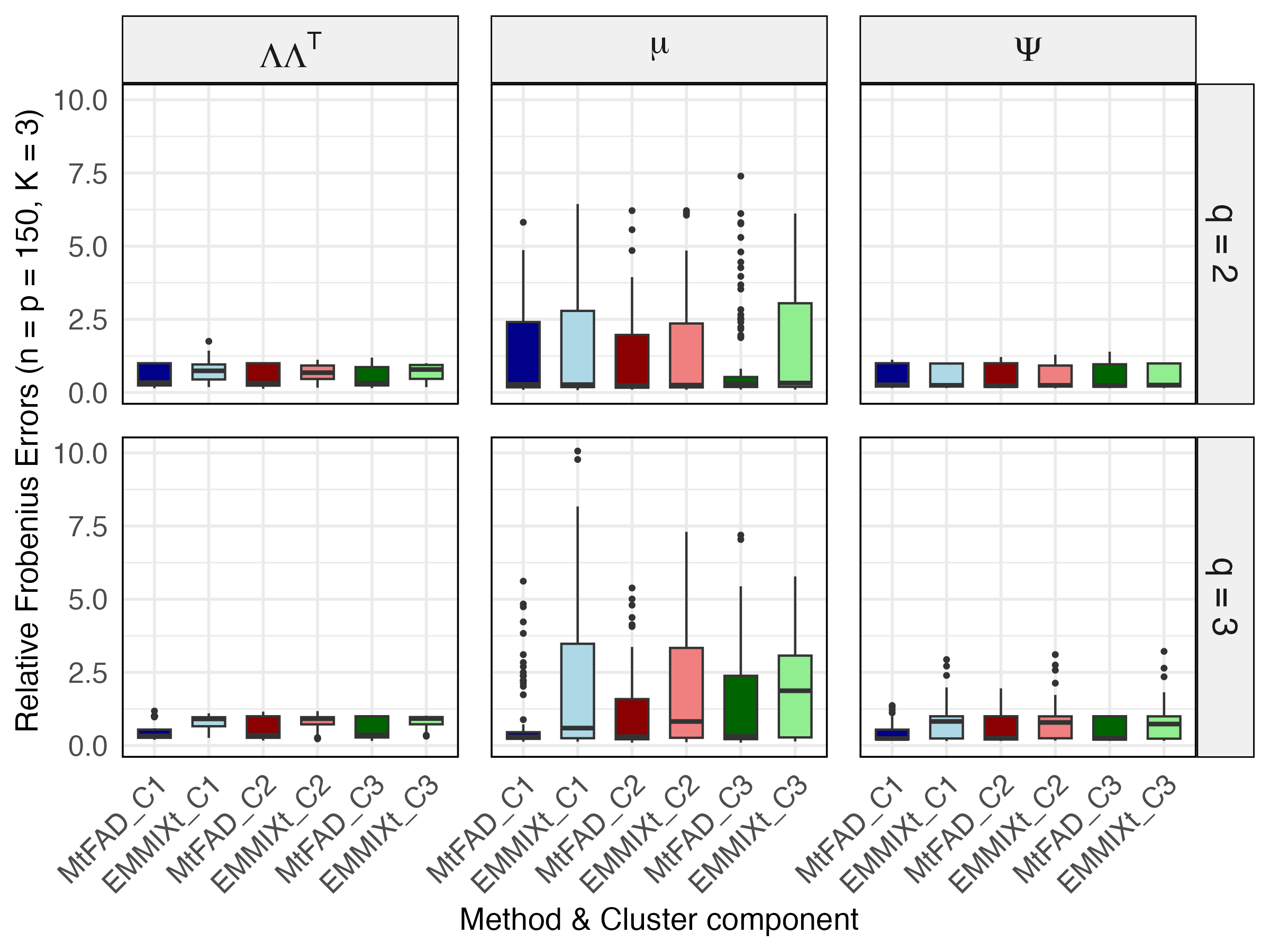}
        \caption*{(b) $K=3$}
    \end{minipage}
    \caption{Relative Frobenius distances of parameters fitted from MtFAD and EMMIXt for  $n=150, p=150$.}
    \label{frobenius-error-p150}
\end{figure}

\section{Additional Results for GRB Data}
\begin{table}[h!]
\centering
%\footnotesize % Or 
\scriptsize %for even smaller font
\setlength{\tabcolsep}{3pt} % Reduce column spacing
\setlength\extrarowheight{-3pt}
\caption{The estimated factor loadings of the five groups for the GRB data fitted with MtFAD-q. For clarity
of presentation, values in the interval $(-0.1,0.1)$ are suppressed in the table.}
\begin{tabular}{|*{6}{c|}} % 5 columns, centered
\hline
Group & Factor 1  & Factor 2 & Factor  3 & Factor 4 & Factor 5 \\ \hline
% Add your 45 rows here   3.42  & 5.61  & 2.87  & 7.33  & 1.92  \\
  \hline
&-0.76 & 0.44 & -0.40 &  \\ 
 & -0.76 & 0.51 & -0.39 & \\ 
 & -0.78 &  & 0.16 & -0.46 \\ 
 & -0.91 &  & 0.25 & -0.31 \\ 
 Group 1 & -0.91 & & 0.21 & 0.34 \\ 
 & -0.33 &  &  & 0.27 \\ 
 & -0.41 & -0.88 & -0.16 &  \\ 
  &-0.41 & -0.89 & -0.20 &  \\ 
  &-0.48 & -0.83 & -0.12 & \\ 
  \hline
  \hline
  & -0.19 & 0.97 & -0.10 &  \\ 
  &-0.21 & 0.86 & &  \\ 
  &-0.83 &  & -0.19 & -0.16 \\ 
  &-0.89 & 0.14 & -0.26 & -0.27 \\ 
  Group 2&-0.96 & 0.16 & 0.20 &  \\ 
  &-0.67 & 0.17 & 0.44 & 0.20 \\ 
  &-0.90 & -0.40 & -0.17 &     \\ 
  &-0.94 & -0.33 & &  \\ 
  &-0.99 &  &  &  \\ 
   \hline
  \hline
&-0.47 & 0.85 &  & 0.21 &  \\ 
 & -0.42 & 0.74 & & &  \\ 
  &-0.94 & 0.15 & 0.26 & -0.13 & -0.10 \\ 
  &-0.97 & 0.14 & 0.15 & -0.10 &  \\ 
  Group 3&-0.97 & 0.16 & -0.10 & & 0.11 \\ 
  &-0.88 & 0.17 & -0.42 &  &  \\ 
  &-0.86 & -0.49 &  & 0.14 & \\ 
  &-0.87 & -0.48 &  & 0.12 &  \\ 
  &-0.89 & -0.44 &  &  & \\ 
  \hline
  \hline
&-0.23 & 0.93 & -0.28 & & \\ 
 & -0.24 & 0.89 & -0.18 & &  \\ 
 & -0.69 & 0.51 & 0.17 & -0.40 & -0.19 \\ 
 & -0.78 & 0.50 & 0.23 & -0.25 &  \\ 
 Group 4& -0.83 & 0.45 & 0.29 &  & \\ 
  &-0.61 & 0.34 & 0.20 & 0.52 & -0.39 \\ 
  &-0.96 & -0.27 & -0.10 &  & \\ 
  &-0.95 & -0.28 & -0.10 & &  \\ 
  &-0.95 & -0.28 & &  &  \\ 
   \hline
  \hline
&-0.17 & -0.79 & -0.30 & 0.32 \\ 
 & -0.24 & -0.79 & -0.29 & 0.28 \\ 
 & -0.83 & -0.30 & -0.18 & \\ 
 & -0.91 & -0.23 & -0.21 & -0.14 \\ 
  Group 5& -0.96 & -0.20 & 0.12 & \\ 
  &-0.66 & -0.14 & 0.41 & 0.31 \\ 
  &-0.96 & 0.27 &  & \\ 
  &-0.99 &  &  &  \\ 
  &-0.98 & -0.18 &  & \\ 
   \hline
\end{tabular}
\label{tab:loadings_grb_varied_q}
\end{table}

\end{document}